\documentclass[12pt,reqno]{article}

\usepackage[usenames]{color}
\usepackage{amssymb}
\usepackage{graphicx}
\usepackage{amscd}
\usepackage{tikz}
\usetikzlibrary{automata,positioning}

\usepackage[colorlinks=true,
linkcolor=webgreen,
filecolor=webbrown,
citecolor=webgreen]{hyperref}

\definecolor{webgreen}{rgb}{0,.5,0}
\definecolor{webbrown}{rgb}{.6,0,0}

\usepackage{color}
\usepackage{fullpage}
\usepackage{float}

\usepackage{graphics,amsmath,amssymb}
\usepackage{amsthm}
\usepackage{amsfonts}
\usepackage{latexsym}
\usepackage{epsf}

\def\orr{ \mid }
\def\divides{{\, | \,}}
\def\Zee{\mathbb{Z}}
\DeclareMathOperator{\ord}{ord}
\DeclareMathOperator{\swm}{swm}
\DeclareMathOperator{\msw}{msw}
\DeclareMathOperator{\mfw}{mfw}

\def\modd#1 #2{#1\ \mbox{\rm (mod}\ #2\mbox{\rm )}}
\DeclareMathOperator{\argmin}{argmin}
\newcommand{\seqnum}[1]{\href{https://oeis.org/#1}{\underline{#1}}}

\title{Computational Aspects of Sturdy and Flimsy Numbers}

\author{Trevor Clokie, Thomas F. Lidbetter\footnote{Author's current address:  RideCo, Waterloo, ON, Canada.}, Antonio Molina Lovett\footnote{Author's current address:  Department of Computer Science, Princeton University, Princeton, NJ  08540, USA.},\\ Jeffrey Shallit, and Leon Witzman\\
School of Computer Science\\
University of Waterloo\\
Waterloo, ON  N2L 3G1 \\
Canada\\
\href{mailto:trevor.clokie@uwaterloo.ca}{\tt trevor.clokie@uwaterloo.ca} \\
\href{mailto:finnlidbetter@gmail.com}{\tt finnlidbetter@gmail.com} \\
\href{mailto:antonio@amolina.ca}{\tt antonio@amolina.ca} \\
\href{mailto:shallit@uwaterloo.ca}{\tt shallit@uwaterloo.ca} \\
\href{mailto:lwitzman@edu.uwaterloo.ca}{\tt lwitzman@edu.uwaterloo.ca}
}

\date{}

\begin{document}

\maketitle

\theoremstyle{plain}
\newtheorem{theorem}{Theorem}
\newtheorem{corollary}[theorem]{Corollary}
\newtheorem{lemma}[theorem]{Lemma}
\newtheorem{proposition}[theorem]{Proposition}

\theoremstyle{definition}
\newtheorem{definition}[theorem]{Definition}
\newtheorem{example}[theorem]{Example}
\newtheorem{conjecture}[theorem]{Conjecture}

\theoremstyle{remark}
\newtheorem{remark}[theorem]{Remark}

\begin{abstract}
Following Stolarsky, we say that a natural number $n$ is {\it flimsy\/} in base $b$ if some positive multiple of $n$ has smaller digit sum in base $b$ than $n$ does; otherwise it is {\it sturdy\/}.  We develop algorithmic methods for the study of sturdy and flimsy numbers.

We provide some criteria for determining whether a number is sturdy.  Focusing on the case of base $b = 2$, we study the computational problem of checking whether a given number is sturdy, giving several algorithms for the problem.    
We find two additional, previously unknown sturdy primes.  We develop a method for determining which numbers with a fixed number of $0$'s in binary are flimsy.  Finally, we develop a method that allows us to estimate the number of $k$-flimsy numbers with $n$ bits, and we provide explicit results for $k = 3$ and $k = 5$.   Our results demonstrate the utility (and fun) of creating algorithms for number theory problems, based on methods of automata theory.
\end{abstract}

\section{Introduction}

Let $s_b (n)$ denote the sum of the digits of $n$, when expressed in 
base $b$.  Thus, for example, $s_2(9) = 2$.   A number $n$ is
said to be {\it $k$-flimsy in base $b$} if there exists a positive
integer $k$ such that $s_b(kn) < s_b(n)$.  
Any such $k$, if one exists, is called
a {\it flimsy witness} for $n$.  If $n$ is $k$-flimsy for some $k$,
it is said to be {\it flimsy}. 
If there is no such
$k$, then $n$ is said to be {\it sturdy in base $b$}.  For example, $7$ is sturdy in base $2$, while $13$ is flimsy, because $s_2 (13) = 3 > 2 = s_2(5\cdot 13)$.  Thus $5$ is a flimsy witness for $13$.   In this paper we examine the computational aspects of sturdy and flimsy numbers.

Sturdy and flimsy numbers were introduced by Stolarsky in 1980 \cite{Stolarsky:1980}.
For other papers on the topic, see \cite{Schmidt:1983,Dartyge&Luca&Stanica:2009,Chen&Hwang&Zacharovas:2014,Basic:2017}.  

Many of the sequences we discuss appear in the {\it On-Line Encyclopedia of Integer Sequences} \cite{Sloane}.  For example, the base-$2$ sturdy numbers form sequence \seqnum{A125121} in the OEIS, while the base-$2$ sturdy primes form sequence \seqnum{A143027}.   The base-$2$ flimsy numbers form sequence
\seqnum{A005360}, while the base-$2$ flimsy primes form sequence
\seqnum{A330696}.
The base-$10$ sturdy numbers form sequence \seqnum{A181862}, while the base-$10$ sturdy primes form sequence
\seqnum{A181863}.   Sequence \seqnum{A086342} gives the value of
$\min_{k \geq 1} s_2(kn)$, while sequence \seqnum{A143069} gives
$\argmin_{k \geq 1} s_2(kn) = \min\{ k \, : \, s_2(kn) = \min_{k \geq 1} s_2 (kn) \}$.

The goal of this paper is to examine the algorithmic aspects of sturdy and flimsy numbers.
The outline of the paper is as follows.
In Section~\ref{basic}, we prove some basic properties of digit
sums of multiples.  In Section~\ref{infinite}, we give a criterion for determining if a number is flimsy, and use it to find two previously unknown sturdy primes.

Next, we turn to algorithms for sturdy and flimsy numbers.  A priori it is not immediately clear that it is even decidable whether a given $n$
is flimsy or sturdy.  Indeed, in a recent paper by 
Elsholtz \cite{Elsholtz:2016}, he asks, ``How can one algorithmically find a `sparse' representation of a multiple of $p$?''  

More precisely, there are four computational problems worthy of study:
\begin{enumerate}
    \item Given a positive integer $n$, decide whether it is sturdy in base $b$.
    \item Compute $\swm_b(n) := \min_{k \geq 1} s_b(kn)$.  This is the smallest digit sum
    of a multiple of $n$; if $n$ is sturdy, then
    $\swm_b(n) = s_b (n)$.
    \item Compute $\msw_b(n) := \argmin_{k \geq 1} s_b(kn)$.   This is the smallest $k$ such that
    $kn$ achieves its minimum digit sum;
    if $n$ is sturdy, then $\msw_b(n) = 1$.
    \item Given that $n$ is flimsy, determine
    $\mfw_b(n) := \min\{k \, : \, s_b(kn) < s_b (n) \}$.   This is the {\it minimal flimsy witness}.
\end{enumerate}
A table of these functions is
given in Appendix C.
In Sections~\ref{algo}--\ref{algo4},
we discuss algorithms to solve these problems.
The fastest, based on automata theory, shows that we can check whether a number $n$ is sturdy in $O(n)$ time.  Section~\ref{results} gives our computational results achieved with our algorithms.

In Section~\ref{moreapp} we give an application of automata to help characterize the flimsy numbers with a fixed number of $0$'s.

Finally, in Section~\ref{flim}, we turn to estimating the number of $k$-flimsy numbers with $n$ bits.  We use techniques from formal language theory to solve the problem.

\section{Basic properties}
\label{basic}

In this section, we prove some of the basic properties of digit sums of multiples.  

We start with some
notation.   For $n \geq 0$, we define $(n)_b$ to be the base-$b$ representation of $n$, starting with the most significant digit.   If $x$ is a string, we define $[x]_b$ to be the integer that $x$ represents when interpreted in base $b$.   If $b$ is fixed, we define $\overline{x}$ to be the base-$b$ complement of $x$, that is, the string where each digit $d$ in $x$ is replaced by $b-1-d$.

\begin{theorem}
Let $b\geq 2$ be an integer, and $t$ be a positive divisor of $b$. Then for all integers $n,r \geq 1$, there exists a positive integer $j$ such that $s_b(jn) = r$ if and only if there exists a positive integer $k$ such that $s_b(ktn) = r$.
\label{thm1}
\end{theorem}
\begin{proof}
For one direction, take $k = j$ and $t = 1$.

For the other direction,
assume that there exists $j \geq 1$ such that $s_b(jn) = r$. Let $k = bj/t$. Then $s_b(ktn) = s_b(bjn) = s_b(jn) = r$.
\end{proof}

We now show that in order
to compute $\swm_b$, it
suffices to consider only
those arguments relatively
prime to $b$.

\begin{corollary}
Write the prime factorization of $n$
as $\prod_{1 \leq i \leq t} p_i^{e_i}$, and define
$g = \prod_{p_i \mid b} p_i^{e_i}$.   Then
$\swm_b(n) = \swm(n/g)$,
and $\gcd(b, n/g) = 1$.
\label{cor2}
\end{corollary}

\begin{proof}
Let $p$ be any prime dividing both $b$ and $n$. From Theorem~\ref{thm1}, we see that $\swm_b (n) = \swm_b(n/p)$.  By repeatedly applying this observation, and replacing $n$ with $n/p$,
we can remove from $n$ all primes dividing both $b$ and $n$,
while maintaining the
same value of $\swm_b$.
At the end, the resulting $n/g$ is relatively prime to $b$. \end{proof}

\begin{theorem}
There exists $j \geq 1$ such that $s_b(jn) = t$ if and only if there exist $t$ \emph{distinct} powers of $b$ that sum to a multiple of $n$.
\label{thm3}
\end{theorem}

\begin{proof}
By Corollary~\ref{cor2}, we may
assume that $n$ is coprime with $b$.

In such cases, $b$ has finite order, say $\nu$, modulo $n$. Suppose $\sum_{i = 0}^{\nu - 1} c_ib^i \equiv \modd {0} {n}$ where each $c_i \geq 0$ and $\sum_{i = 0}^{\nu - 1} c_i = t$. Then $\sum_{i = 0}^{\nu - 1} \sum_{j = 0}^{c_i - 1}b^{j\nu + i} \equiv \modd {0} {n}$, and this sum consists of distinct powers of $b$.
\end{proof}

Empirical evidence suggests that
if $b = 2$ and
$\swm_b (n) = t$, then for all $i \geq 0$, some multiple of $n$ has digit sum $t+i$.   However, the analogous result
is false for $b = 3$.  For example,
$\swm_3 (13) = 3$, but no multiple of $13$ has digit sum $4$.  These observations are explained in the following theorem.

\begin{theorem}
Suppose $j, n$ are positive integers such that $s_b(jn) = t$. Then for all $r \geq 0$, there exists $k \geq 1$ such that $s_b(kn) = t + r(b - 1)$.
\end{theorem}

\begin{proof}
Assume $s_b(jn) = t$ for some $t \geq 1$.   Then from Theorem~\ref{thm3} we know that $\sum_{i=1}^t b^{m_i} \equiv \modd{0} {n}$ for some strictly increasing $m_i$ and (replacing $j$ by $bj$ if needed) we can assume
$m_t \geq 1$.  Now replace the high-order bit $b^{m_t}$ in this sum with
the sum of $b$ terms $b^{\nu + m_t - 1} +
b^{2\nu + m_t - 1} + \cdots + 
b^{(b-1) \nu + m_t - 1}$,
where $\nu$ is the order of $b$, modulo $n$.  This has the effect removing $1$ bit, while adding $b$ additional bits, and each of the $b$ new terms is congruent to $b^{m_t - 1}$ (mod $n$).  So we have found another multiple of $n$ with digit sum $t +b-1$.   We can repeat this transformation
any number of times.
\end{proof}

\section{Infinite classes of sturdy numbers}
\label{infinite}

We first give a criterion for deciding whether a number is flimsy. This shows that Problem 1 on our list, determining whether a given positive integer is sturdy, is decidable.

\begin{theorem}
Let $n,b,j$ be positive integers, $b \geq 2$ such that $n$ divides $b^j - 1$. Then $n$ is flimsy in base $b$ if and only if $s_b(kn) < s_b(n)$ for some $k$ satisfying $1 \leq k \leq \frac{b^j-1}{n}$.
\label{thm5}
\end{theorem}

\begin{proof}
One direction is easy:  if $s_b(kn) < s_b(n)$ for some $k$, then $n$ is flimsy in base $b$. 

For the other direction, suppose $n$ is
flimsy, but
$s_b(kn) \geq s_b(n)$ for all $k$
with $1\leq k \leq \frac{b^j-1}{n}$. Let $k'$ be the smallest positive integer such that $s_b(k'n) < s_b(n)$. By assumption $k'n \geq b^j$, and so we can write $k'n = cb^j + d$ for uniquely-determined $c \geq 1$ and $0 \leq d < b^j$. Since $b^j \equiv \modd{1} {n}$, it follows that $cb^j + d \equiv c + d \equiv \modd{0} {n}$. Then $c + d = fn < cb^j + d = k'n$ for some integer $f$ with $1 \leq f < k'$. Thus $s_b(k'n) = s_b(cb^j + d) = s_b(c) + s_b(d) \geq s_b(c + d) = s_b(fn) \geq s_b(n)$, achieving the desired contradiction.
\end{proof}

\begin{remark}
Since $j \leq \varphi(n)$, this together with Theorem~\ref{thm1} shows that sturdiness is reduced to a finite search.  The result for $b = 10$ was observed by Phedotov \cite{Phedotov:2002}.

We applied Theorem~\ref{thm5} to known prime factors of composite Mersenne numbers
\cite{Wagstaff} and found  $$57912614113275649087721 =   \frac{2^{83} - 1}{167}$$
and 
$$10350794431055162386718619237468234569 = \frac{2^{131} - 1}{263}$$
as previously unknown sturdy primes in base $2$.
\end{remark}

\begin{corollary} If $b,j$ are positive integers, with $b \geq 2$, then
$\frac{b^j - 1}{m}$ is sturdy in base $b$ for every positive $m$ dividing $b - 1$.
\label{cor6}
\end{corollary}

\begin{proof}
Let $k$ be an integer with $1 \leq k \leq m$.  Then we have
$$s_b \left(k\frac{b^j-1}{m} \right) = s_b\left( \frac{k(b-1)} {m} \sum_{i=0}^{j-1} b^i \right) = kj\frac{b-1}{m} \geq j\frac{b-1}{m} ,$$
where we have used the fact that $k \leq m$.
The result now follows by Theorem~\ref{thm5}.
\end{proof}

We can also get a generalization
of a theorem of Stolarsky
\cite[Thm.~2.1]{Stolarsky:1980}.

\begin{corollary}
Let $b\geq 2$, and let $r, e$ be integers.   Define
$n = {{b^{re}- 1} \over {b^e - 1}}$.
Then $n$ is sturdy in base $b$.
\end{corollary}

\begin{proof}
Observe that $(n)_b = (1\, 0^{e-1})^{r-1}\, 1$, and so
$s_b(n) = r$.  On the other hand,
if $1 \leq k \leq b^{e-1}$, then
$(kn)_b$ consists of $r$ copies of
$(k)_b$, concatenated, separated by
some number of $0$'s.   So
$s_b(kn) = r s_b (k) \geq r$.
The result now follows by Theorem~\ref{thm5}.
\end{proof}

The next theorem gives an infinite class of sturdy numbers.

\begin{theorem}Fix $b \geq 2$. Let $n$ be a positive integer, and $x$ be the base-$b$ representation of $n$. Then every integer with base-$b$ representation of the form $x \, (b-1)^i \, \overline{x}$, where $i \geq 0$ and $\overline{x}$ is the base-$b$ complement of $x$, is sturdy in base $b$. 
\label{thm9}
\end{theorem}

\begin{proof}
Suppose $y = x \, (b-1)^i \, \overline{x}$ for some $i \geq 0$. Then $[y]_b + n = nb^{\lvert x \rvert + i} + b^{\lvert x \rvert + i} - 1$. Then $[y]_b = (n+1)(b^{\lvert x \rvert + i} - 1)$. Observe that $s_b(b^{\lvert x \rvert + i} - 1) = (\lvert x \rvert + i)(b-1) = s_b([y]_b)$. Then for every positive integer $k$ we have $s_b(k[y]_b) = s_b(k(n+1)(b^{\lvert x \rvert + i} - 1)) \geq s_b(b^{\lvert x \rvert + i} - 1) = s_b([y]_b)$.  By Corollary~\ref{cor6} it follows that By Corollary~\ref{cor6} it follows that $b^{|x|+i}-1$ is sturdy in base $b$, and hence $y$ is.
\end{proof}

\begin{corollary}
Let $b \geq 2$ be an integer, and $m$ be an integer such that $m^2$ divides $b - 1$. Then $\frac{(b^n - 1)^2}{m^2}$ is sturdy in base $b$ for all $n \geq 1$.
\end{corollary}

\begin{proof}
Suppose $m^2$ divides $b-1$. Then we have
\begin{align*}
\frac{(b^n - 1)^2}{m^2} &= \frac{b^n - 1}{m^2}b^n - \frac{(b^n - 1)}{m^2}\\
&= \frac{b^n - 1}{m^2}b^n - b^n + b^n - \frac{(b^n - 1)}{m^2}\\
&= \left(\frac{b^n - 1}{m^2} - 1\right)b^n + b^n - \frac{(b^n - 1)}{m^2}\\
&= \left(\frac{b^n - 1}{m^2} - 1\right)b^n + (b^n - 1) - \left(\frac{(b^n - 1)}{m^2} - 1\right),
\end{align*}
which has base-$b$ representation $x\overline{x}$ where $[x]_b = \frac{b^n - 1}{m^2} - 1$. Then by Theorem~\ref{thm9}, the number
$\frac{(b^n - 1)^2}{m^2}$ is sturdy.
\end{proof}

In the rest of this paper
we are almost exclusively concerned with the case $b = 2$, and so from now on we omit the subscripts on the functions $\msw, \swm, \mfw$,
and use the terms
{\it flimsy\/} or {\it sturdy\/} without further elaboration. In this case $s_2(n)$ equals the number of $1$'s in the binary representation of $n$, also known as the {\it Hamming weight} of $n$.

\begin{theorem} Let $r,e$ be positive integers.  Every integer with base-$2$ representation $(1^e0^e)^r1^e$ is sturdy in base $2$.
\end{theorem}

\begin{proof} Let $n = [(1^e0^e)^r1^e]_2 = \sum_{i=0}^{r-1} (2^e-1)2^{2ie} = \frac{2^{2re} - 1}{2^e + 1} = -1 + 2^e - 2^{2e} + \cdots + 2^{(2r-1)e}$. We must show $s_2(kn) \geq s_2(n)$ for all positive integers $k$. We proceed by induction on $k$. The statement is clearly true for all $k \leq 2^e + 1$. Consider $k \geq 2^e + 2$. Then there exist uniquely-defined integers $s,t$ such that $k = 2^es + t$ where $s \geq 1$ and $1 \leq t < 2^e$. When $t$ is even, we have $s_2(kn) = s_2(\frac{k}{2}n) \geq s_2(n)$ by the inductive hypothesis. Now assume $t = 2t' + 1$ for some $t' \geq 0$.

\medskip

\noindent Case 1:  $s$ odd.
Write $s = 2s' + 1$ for some $s' \geq 0$. Using the properties that $s_2(2a + 1) = s_2(a) + 1$ and $s_2(a + 2^b - 2^b) \geq s_2(a + 2^b) - 1$ for all integers $a,b \geq 0$ \cite{Stolarsky:1980}, we have
\begin{align*}
s_2(kn) &= s_2((2^es + t)n)\\
&= s_2((s+t)n + (2^e + 1)sn - 2sn)\\
&= s_2((s+t)n + 2s'(2^{er}-1) + (2 + 2^2 + \cdots + 2^{er - 1} + 2^{er}) - 2sn - 2^{er}) + 1\\
&\geq s_2((s+t)n + 2s'(2^{er}-1) + 2(2^{er} - 1) - 2sn)\\
&= s_2((s'+t'+1)n + (s'+1)(2^e+1)n - sn)\\
&= s_2((t' + (s'+1)2^e + 2)n) \geq s_2(n)
\end{align*}
by the inductive hypothesis, since $t' + (s'+1)2^e + 2 < 2^es + t$ when $k \geq 2^e + 2$. 

\medskip

\noindent Case 2:  $s$ even.  Write $s = 2s'$ for some $s' \geq 1$. We have
\begin{align*}
s_2(kn) &= s_2((2^es + t)n)\\
&= s_2(2^esn + 2t'n + 2n - n)\\
&= s_2(2^esn + 2t'n + 2n + (1 - 2^e + 2^{2e} - \cdots - 2^{(2r-1)e}))\\
&= s_2(2^esn + 2t'n + 2n + 2^e(-1 + 2^e - 2^{2e} + \cdots + 2^{(2r - 1)e}) - 2^{2re}) + 1\\
&\geq s_2(2^esn + 2t'n + 2n + 2^en)\\
&= s_2((2^es' + 2^{e-1} + t' + 1)n) \geq s_2(n)
\end{align*}
by the inductive hypothesis, since $2^es' + 2^{e-1} + t' + 1 < 2^es + t$.  This completes the proof.
\end{proof}

\section{Algorithms when $\swm(n)$ is small}
\label{algo}

As we will see in Section~\ref{algo2}, for general $n$ we can determine whether $n$ is sturdy in $O(n)$ time.   We call this a linear-time algorithm.\footnote{Strictly speaking, the usage ``linear-time'' in the context of algorithms on integers would usually mean an algorithm that runs in $O(\log n)$ time.  But since no algorithm is this efficient, we stray from the common usage for brevity.}   Therefore, it is of interest to see when this can be improved.

If $\swm(n)$ is small, this fact can be verified  efficiently in some cases.  This is particularly relevant in the case where $n$ is prime because, according to a recent result of Elsholtz \cite{Elsholtz:2016}, almost all primes $p$ have $\swm(p) \leq 7$.   Furthermore, we know from results of Hasse \cite{Hasse:1966} and 
Odoni \cite{Odoni:1981} that a positive proportion of all primes satisfy $\swm(p) = 2$;
asymptotically, this fraction is $17/24$. For general $n$, however, the situation is different:  the set of $n$ for which
$\swm(n) = 2$ has density $0$; see the results of Moree in
\cite[Appendix B]{Pless&Sole&Qian:1997}.

\subsection{The case $\swm(n) = 2$}
\label{babystep}

If $\swm(n) = 2$, then 
$n\cdot\msw(n) = 2^k + 1$ for some integer $k \geq 1$.  Hence $n \divides 2^k + 1$, and so $-1$ 
belongs to the subgroup generated by $2$ (mod $n$). We can decide if $-1$ belongs to the 
subgroup generated by $2$ (mod $n$) by using an algorithm for the discrete logarithm problem. 
For example, the baby-step giant-step algorithm can be used to find $k$ such that 
$2^k \equiv -1 \pmod{n}$, if such a $k$ exists, with time complexity $O(\sqrt{n}\log n)$ 
\cite{Shanks:1969}.   If the factorization of $n$ is known, this running time can be substantially improved.

\subsection{The case $\swm(n) = 3$}

If $\swm(n) = 3$, then 
$n\cdot\msw(n) = 2^k + 2^\ell + 1$ for
some integers $k > l \geq 1$.
It follows that $2^k + 2^l \equiv \modd{-1} {n}$, which means that we are dealing with a 2-SUM problem. This can be solved in
$O(n \log n)$ time using sorting and binary search.  (Briefly, compute a table of powers of $2$, mod $n$; sort them in ascending order, and then for each power $2^k$ use binary search to see if there is an $\ell$ such that $2^l \equiv \modd{-1-2^k} {n}$.)
Although this does not beat our $O(n)$ algorithm given below asymptotically, in many cases it will run more quickly because of the simplicity of the operations.  This is particularly true if the subgroup generated by $2$ (mod $n$) is small.

\section{A dynamic programming algorithm}
\label{algo2}

In this section we show how to check whether $n$ is sturdy using dynamic programming.

By Corollary~\ref{cor2}, we can restrict our attention to the
case where $n$ is odd.  In this case,
the powers of two $P_n = \{ 2^i \, : \, i \geq 0 \}$
form a cyclic subgroup of $(\Zee/(n))^*$, the multiplicative group of integers relatively prime to $n$.  Define $\nu = \ord_2 n = |P_n|$, the order of $2$ in the group $(\Zee/(n))^*$.
Hence, to find a positive multiple of
$n$ whose binary expansion contains exactly $k$ $1$'s, it suffices to find
an appropriate linear combination of $k$ elements of $P_n$ (counted with
repetition)  that sums to
$0$ (mod $n$).  More precisely, we need to find non-negative integers
$a_1, a_2, \ldots, a_i$ and distinct elements $e_1, e_2, \ldots, e_i \in P_n$
such that 
\begin{align*}
a_1 e_1 + \cdots + a_i e_i &\equiv \modd{0} {n} \\
a_1 + a_2 + \cdots + a_i &= k ,
\end{align*}
for integers $k \geq 1$.  
This is the kind of problem that dynamic programming is well-suited for. To restrict the amount of work required in a dynamic programming algorithm for this we make use of the following lemma.

\begin{lemma}
\label{lem:distinctsummands}
For an integer base $b\geq 2$ let $P_{b,n}=\{b^i \bmod n : i \in\mathbb{N} \}$ and suppose that $e_1, e_2, \dots, e_m$ are the distinct elements of $P_{b,n}$. If there exist non-negative integers $a_1, a_2, \dots, a_m$ such that $a_1e_1 + a_2e_2 + \cdots + a_me_m \equiv \modd{0} {n}$ and $a_1+\cdots + a_m=k$, then there exist non-negative integers $c_1, c_2, \dots, c_m < b$ and $l\leq k$ such that $c_1e_1 + c_2e_2 + \cdots + c_me_m \equiv \modd{0} {n}$ and $c_1 + \cdots + c_m = l$.
\end{lemma}
\begin{proof}
Suppose we have non-negative integers $a_1, a_2, \dots, a_m$ such that $a_1e_1 + a_2e_2 + \cdots + a_me_m \equiv \modd{0} {n}$ and $a_1+\cdots + a_m=k$. If we have $a_1, a_2, \dots, a_m < b$ then we are done. So suppose that there is some $i$ such that $a_i\geq b$. Let $j$ be the integer such that $be_i \equiv \modd{e_j} {n}$. 
Then we can take 
$$\left(\sum_{r=1, r\neq i, r\neq j}^{m}a_re_r\right) + (a_i - b)e_i + (a_j + 1)e_j \equiv \modd{0} {n},$$ giving $$\left(\sum_{r=1, r\neq i, r\neq j}^m a_r \right) + (a_i - b) + (a_j + 1) = a_1 + \cdots + a_m - b + 1 = k - b + 1 < k.$$ 
Setting $a_i:= a_i - b$ and $a_j := a_j + 1$ and $k:= k-b+1$, we can repeat this argument until $a_1, \dots, a_m < b$.
\end{proof}

Let us start with determining whether $n$ is sturdy. It suffices to solve the problem of the previous paragraph for $1 \leq k < s_2 (n)$.    The idea is that we will fill in the entries of a $3$-dimensional boolean array
$x$ with the following meaning:  the entry
$x[i,j,r]$ is {\tt true} if and only if
the integer $j$ has a representation as a sum of $i \geq 1$ powers of $2$, using as summands only the first $r$ elements of the set $P_n$ without repetition.    We fill in the array $x$ in increasing order of $r$.

For initialization, we set all elements
of $x$ to {\tt false}, except that we set 
$x[0,0,r]$ to {\tt true} for $0 \leq r \leq \nu$.

To solve the remaining three problems, we need to
record more information than just the ability
to represent $j$ as a sum of powers of $2$.
The integer array $y[i,j,r]$ is used to record the smallest integer congruent to $j$ (mod $n$) that is the sum of exactly $i$
powers of $2$ (without repetition), using only
the first $r$ elements of the set $P_n$.

The complete algorithm is given in Appendix~\ref{appa}.

Our dynamic programming algorithm has three nested loops,
which gives a running time of $O(\nu \cdot n \cdot s_2 (n))$.  Since $\nu = \ord_n 2$ could be as large as $n-1$, and $s_2(n)$ could be as big
as $\log_2 n$, this gives a worst-case running time of $O(n^2 \log n)$, where we are measuring the run time in terms of RAM operations on integers of size about $n$.
This means that this algorithm will only
be feasible for integers smaller than about
$10^7$.





\section{An algorithm based on finite automata}
\label{algo3}

In this section we provide a different,
much faster algorithm for checking sturdiness, based on finite automata.

The idea is simple.  It is easy to create a deterministic
finite automaton (DFA) accepting the binary representations
of the positive integers divisible by $n$.   Such an automaton has
$n$ states \cite{Alexeev:2004} and exactly one final state.  Next, we can easily construct a DFA $A_t$ accepting those strings
starting with a $1$ and having at most $t$ ones. 
Using the standard ``direct product''
construction \cite[pp.~59--60]{Hopcroft&Ullman:1979}, we can construct a DFA $M_t$ of $(t+2)n$ states for the intersection of these
two languages; it has exactly $t+1$ final states $f_0, f_1, \dots, f_t$ corresponding to positive integers divisible by $n$ with $0,1,\dots, t$ $1$'s respectively.  Then some multiple of $n$ has at most $t$ $1$'s iff $M_t$ accepts
at least one string.  We can test this condition (and even find the lexicographically
least string accepted) using breadth-first search to decide if some $f_i$ for $0\leq i\leq t$ is reachable from the
start state of $M_t$, in linear time
in the size of $M$, so in $O((t+2)n)$ time.

By choosing $t=s_2(n) - 1$ we can determine if $n$ is sturdy in $O(n \log n)$ steps.
Similarly, by allowing the breadth-first search to run to completion and keeping track
of the least string in radix order used to reach each state, we can recover $\swm(n)$,
$\msw(n)$, and $\mfw(n)$ by examining each of the final states for whether or not they
were visited in the search and looking at the least string in radix order used in each
case. More precisely, the value of $\swm(n)$ is the least integer $i$ such that final state 
$f_i$ in $M_t$ is reached in the breadth-first search, or $s_2(n)$ if no final state is reached.
The value of $\msw(n)$ is the least string in radix order used to reach $f_{\swm(n)}$ 
interpreted as an integer and divided by $n$, or $1$ if $n$ is sturdy. The value of $\mfw(n)$, if
it is defined, is the least string in radix order among all such strings used to reach a
final state, interpreted as an integer and divided by $n$.   To avoid needing to store the representation of large integers, we instead store the exponents of the current power of $2$ and a pointer to the previous power.  From this linked list we can reconstruct the appropriate number.


\begin{theorem}
We can decide whether $n$ is sturdy $O(n \log n)$ steps.  In the same time
bound we can compute $\swm(n)$ and $\msw(n)$.
If $n$ is flimsy, we can compute $\mfw(n)$ in
the same time bound.
\end{theorem}

This algorithm is practical for $n$ up to about
$10^{10}$.  The main constraint is likely to be space and not time.  

\section{Improving the automaton-based algorithm}
\label{algo4}

With a small modification to this idea of using a breadth-first search on the graph defined by automaton $M$, we can make further improvements to the time complexity. Consider the deterministic finite automaton $M_n$ accepting the binary representations of the positive integers divisible by $n$. We then define a directed graph $G_n$ with vertices given by the states of $M_n$ and directed, weighted edges given by the transitions of $M_n$ where transitions on the symbol 0 are given an edge weight of 0 in $G_n$ and transitions on the symbol 1 are given an edge weight of 1 in $G_n$. We augment $G_n$ with one additional vertex, $v_s$, with a single outgoing edge of weight 1 to the vertex corresponding to the state reached when $M_n$ reads any input of the form $0^*1$. If $v_f$ is the vertex corresponding to the accepting state in $M_n$, then there is a path from $v_s$ to $v_f$ of weight $k$ if and only if there is a non-zero multiple of $n$ with Hamming weight $k$. The shortest path problem on a graph $G=(V,E)$ with edge weights in $\{0,1\}$ can be solved in time $O(|V|+|E|)$ using a variation of the breadth-first search algorithm. In place of the queue used in a standard breadth-first search, we use a double ended queue. We process a node by traversing incident edges of weight 0 and pushing the nodes reached to the front of the queue if they have not been processed already. Edges of weight 1 are also traversed, but the nodes reached are pushed to the back of the queue provided that they have not been processed already. After a node has been processed, the next node at the front of the queue is dequeued and processed if it has not been processed already, otherwise it is just discarded. The depth of the search can be tracked as in a standard breadth-first search. Thus we achieve the following improvement.

\begin{theorem}
We can test if $n$ is sturdy in $O(n)$ steps.
\end{theorem}

From this approach we are still able to construct an example of a multiple of $n$ achieving the minimum Hamming weight over all multiples of $n$. It is simply a matter of maintaining the path used in the breadth-first search algorithm finding the shortest path from $v_s$ to $v_f$ in $G_n$. However, there is no guarantee that this is the least multiple of $n$ with this property. To find the least multiple we can use the linear-time algorithm to first determine the minimum Hamming weight. For minimum Hamming weight $k$, we take the direct product of automaton $M_n$ accepting the base-2 representations of all multiples of $n$ and the automaton accepting all binary strings with exactly $k$ 1's. A breadth-first search on this product automaton gives the least non-zero multiple of $n$ with Hamming weight $k$. This second breadth-first search has worst case time complexity $O(n \log n)$, giving overall complexity $O(n \log n)$ for finding the least non-zero multiple of $n$ having the minimum Hamming weight over all non-zero multiples of $n$.

\section{Another breadth-first search approach}
\label{algo5}
We can take advantage of Lemma~\ref{lem:distinctsummands} to evaluate sturdiness and compute $\swm$ and $\msw$ using a breadth-first search on a different graph structure. As before, to test the sturdiness of an integer $n\geq 3$, we construct an $(n+1)$-vertex graph with $n$ of the vertices representing the distinct residue classes modulo $n$, which we will refer to as $[0],[1],\dots,[n-1]$, and one special vertex, $v_0$, corresponding to the number 0. The graph contains a directed edge from vertex $[x]$ to vertex $[y]$ if and only if $x+2^j \equiv \modd{y} {n}$ for some integer $j\geq 0$. Similarly, there is an edge from $v_0$ to $[y]$ if and only if $2^j\equiv \modd{y} {n}$. Hence each vertex has out-degree $\nu = \ord_n 2$. 
The idea of this construction is to treat traversing an edge from $[x]$ to $[x+2^j]$ as choosing to use the $j$th power of 2 as a summand in a summation to a value congruent to 0 modulo $n$. Thus, to compute $\swm(n)$ we are looking for the length of the shortest path from $v_0$ to $[0]$ and this can be found via a breadth-first search. Furthermore, by keeping track of the smallest sum required to reach each state, we can also recover $\msw(n)$ from such a breadth-first search. Rather than running the breadth-first search to completion, we can terminate as soon as we reach a depth equal to $s_2(n)$, as we will know by then whether or not $n$ is sturdy.

With this approach we can take advantage of the structure of the graph to speed up testing for sturdiness. If during the breadth first search we visit a node $[x]$ such that $[n-x]$ has already been visited, then since the length of the shortest path from $[x]$ to $[0]$ is equal to the length of the shortest path from $v_0$ to $[n-x]$ either we will know that $n$ is not sturdy, or that it is not necessary to continue searching from $[x]$. This greatly, improves the efficiency of the testing for sturdiness.

The complexity of this approach, for evaluating sturdiness, and computing $\swm(n)$ and $\msw(n)$ is $O(n^2)$ since we are performing a breadth-first search on a graph with $n+1$ vertices each with $\nu = O(n)$ outgoing edges. In practice this approach seems to perform much better than our naive $O(n^2)$ upper bound would suggest, especially in testing for sturdiness, due to the early exit conditions. 

\section{Running time comparison}
\label{running-time}

To demonstrate how these algorithms behave in practice, we compiled timing information for each of the approaches and each of the four functions of interest for consecutive integers starting from $1$. Each of the algorithms are implemented as described above. However, before applying each algorithm the baby-step giant-step algorithm, as in Section~\ref{babystep}, is used to exit faster in those cases where $\swm(n)=2$.
Running times for the $\mfw$ function with the \texttt{order\_deg\_bfs} algorithm are not given because there does not seem to be a natural approach for using this idea to evaluate $\mfw$.
The computations producing the given running times were performed on macOS Catalina version 10.15.2 on a 2.3 GHz Intel Core i5 processor. Implementations of our algorithms can be found in the Github repository\\
\centerline{{\tt  https://github.com/FinnLidbetter/sturdy-numbers } \ \ .}

In the tables below, the algorithmic approaches are named according to the commands used in the program-runner in the Github repository. Here, the \texttt{dp} algorithm refers to the dynamic programming approach described in Section~\ref{algo2}, the \texttt{aut} algorithm refers to the automaton-based approach described in Section~\ref{algo3}, the \texttt{bfs01} algorithm refers to the improved automaton-based approach described in Section~\ref{algo4}, and the \texttt{order\_deg\_bfs} algorithm refers to the alternative breadth-first search approach described in Section~\ref{algo5}.
\begin{table}[H]
    \centering
    \begin{tabular}{|c|c|c|c|c|}
    \hline
    Algorithm                & \texttt{is\_sturdy} & $\swm$  & $\msw$  & $\mfw$  \\ \hline
    \texttt{dp}              & 22836               & 2667063 & 5675228 & 5167556 \\ \hline
    \texttt{aut}             & 1042                & 1050    & 1473    & 1430    \\ \hline
    \texttt{order\_deg\_bfs} & 322                 & 1646    & 4339    & ---     \\ \hline
    \texttt{bfs01}           & 224                 & 226     & 650     & 1416    \\ \hline
    \end{tabular}
    \caption{Running time in milliseconds for each of the algorithms to evaluate the functions for all values of $n$ (odd and even) between $1$ and $2000$ inclusive.}
    \label{tab:running_time_1-2000}
\end{table}

\begin{table}[H]
    \centering
    \begin{tabular}{|c|c|c|c|c|}
    \hline
    Algorithm                & \texttt{is\_sturdy} & $\swm$  & $\msw$ & $\mfw$ \\ \hline
    \texttt{aut}             & 31950               & 31990   & 42229  & 41747  \\ \hline
    \texttt{order\_deg\_bfs} & 7378                & 164543  & 439794 & ---    \\ \hline
    \texttt{bfs01}           & 5209                & 5207    & 15515  & 41761  \\ \hline
    \end{tabular}
    \caption{Running time in milliseconds for the algorithms to evaluate the functions for all values of $n$ (odd and even) between $1$ and $10000$ inclusive. The dynamic programming algorithm was not included because it was not feasible to evaluate the functions for all integers between $1$ and $10000$ with this approach.}
    \label{tab:running_time_1-10000}
\end{table}

\section{Computational results}
\label{results}

 Sequence \seqnum{A143027} in the OEIS \cite{Sloane} gives a list of the first few sturdy primes, namely,
 $$2, 3, 5, 7, 17, 31, 73, 89, 127, 257, 1801,
 2089, 8191, 65537, 131071, 178481, 262657, 524287, 2099863,$$
and mentions $616318177$ as an additional sturdy prime, although it was not known if this was the next sturdy prime to occur in the sequence.
Using our methods, we checked all primes
$p < 2^{32}$.  We confirmed the results in the OEIS and 
found that $616318177$ and $2147483647$ are the only remaining sturdy primes in that range.  The computation took approximately 23 hours on a laptop.

We also computed frequency counts for
the values of $\swm$ for odd $n>1$, not just primes, and they are given in Table~\ref{tab1}.
\begin{table}[H]
\begin{center}
\begin{tabular}{ *{5}{|c}| } 
 \hline
 $\swm(n)$ & $n < 2^{20}$ & $2^{20} < n < 2^{21}$ & $2^{21} < n < 2^{22}$ & $2^{22} < n < 2^{23}$ \\
 \hline
 2  & 115931 & 107650 & 208333 & 403823  \\ \hline
 3  & 286681 & 294938 & 596522 & 1205753 \\ \hline
 4  & 83895  & 83958  & 168138 & 336448  \\ \hline
 5  & 19287  & 19242  & 38566  & 77071   \\ \hline
 6  & 9903   & 9892   & 19812  & 39635   \\ \hline
 7  & 4246   & 4265   & 8510   & 17023   \\ \hline
 8  & 2274   & 2269   & 4548   & 9104    \\ \hline
 9  & 1027   & 1030   & 2058   & 4119    \\ \hline
 10 & 529    & 527    & 1059   & 2118    \\ \hline
 11 & 256    & 257    & 514    & 1024    \\ \hline
 12 & 130    & 131    & 260    & 521     \\ \hline
 13 & 64     & 64     & 128    & 256     \\ \hline
 14 & 32     & 33     & 64     & 129     \\ \hline
 15 & 16     & 16     & 32     & 64      \\ \hline
 16 & 8      & 8      & 16     & 32      \\ \hline
 17 & 4      & 4      & 8      & 16      \\ \hline
 18 & 2      & 2      & 4      & 8       \\ \hline
 19 & 1      & 1      & 2      & 4       \\ \hline
 20 & 1      & 0      & 1      & 2       \\ \hline
 21 & 0      & 1      & 0      & 1       \\ \hline
 22 & 0      & 0      & 1      & 0       \\ \hline
 23 & 0      & 0      & 0      & 1       \\ \hline
\end{tabular}
\end{center}
\label{tab1}
\caption{Counts of $\swm$.}
\end{table}

We also computed counts of sturdy numbers up to $10^i$ for $i = 1, 2,3, 4, 5, 6$, and they are given below:
\begin{table}[H]
    \centering
    \begin{tabular}{c|c}
        $i$ & Number of sturdy numbers $<10^i$ \\
        \hline
         1 & 5 \\
         2 & 22 \\
         3 & 81 \\
         4 & 292 \\
         5 & 995 \\
         6 & 3438
    \end{tabular}
    \caption{Counts of sturdy numbers}
    \label{sturdycount}
\end{table}

\section{Numbers with few $0$'s}
\label{moreapp}

We can also use finite automata to determine when numbers with few 0's are flimsy.  More precisely, for each pair of integers $ j, k$ we can build a DFA $M_2(j,k)$ accepting those $(n)_2$ for
which $(n)_2$ has $j$ 0's and $(kn)_2$ has more than  $j+t$ 0's, where
$t = | (kn)_2 | - | (n)_2 |$.  Such an $n$ is guaranteed to be flimsy.  We can determine $t$ by reading the input $n$,
least significant digit first, and computing $(kn)_2$ on the fly, keeping track of the carries.

Let $j$ be a fixed natural number.  By choosing an appropriate set of flimsy witnesses $k$ (which can be guessed empirically), we can determine all flimsy numbers having exactly $j$ $0$'s in their binary representation.   We do this by computing the DFA's $M_2(j,k)$ and unioning them together to get a final automaton $M'_j$. We expect there to be a finite set of ``sporadic'' sturdy exceptions,  and (according to Theorem~\ref{thm9}) an infinite set of sturdy exceptions consisting of those numbers with binary representation of the form $s 1^i \overline{s}$, where $s$ begins with $1$ and ends with $0$.  This expectation can then be verified by
considering the language accepted by $M'_j$; the finite set of sturdy exceptions can be tested using our algorithms previously discussed.  The multipliers we used in constructing $M'_j$ are the odd numbers $\leq 2^{j+1} + 1$.

With these ideas we can prove
the following theorem.
\begin{theorem}
\leavevmode
\begin{itemize}
    \item[(a)]  Every integer with no $0$'s is sturdy.
    \item[(b)] Every odd integer with one $0$
    is flimsy, with the exception of $5 = [101]_2$, and is proven flimsy by multiplier $3$ or $5$.
    \item[(c)] Every odd integer with two $0$'s is flimsy,
    with the exception of
    $51$ and numbers of the form
    $1 0 1^i 0 1$, $i \geq 0$, which
    are all sturdy.
    \item[(d)]  Every
    odd integer with three $0$'s is flimsy, with the exception of
    $17, 85, 89, 455$ and numbers of the form
    $1001^i 011$ or
    $1101^i001$, $i \geq 0$, which are
    all sturdy.
    \item[(e)] Every odd
    integer with four $0$'s is flimsy,
    with the exception of
    $33, 69, 73, 153, 3855$,
    and numbers of the
    form $1000 1^i 0111$,
    $1100 1^i 0011$,
    $1010 1^i 0101$,
    $1110 1^i 0001$,
    $i \geq 0$,
    which are all sturdy.
    \item[(f)] Every odd integer with five $0$'s is flimsy, with the exception of $65, 133, 161, 267, 275, 1365$,
    $31775$, and numbers specified by
    Theorem~\ref{thm9}.
    \item[(g)] Every odd integer with six $0$'s is flimsy, with the exception of $129, 259, 261, 273, 385, 525$,\\
    $549, 561, 585, 645, 657, 705, 771, 777, 801, 1729, 1801, 2275, 3185, 11565, 13107, 258111$, and numbers specified by
    Theorem~\ref{thm9}.
    \item[(h)] Every odd integer with seven $0$'s is flimsy, with the exception of $257,515,517,529,1035,1065$,\\
    $1105,1155,1157,1185,1285,1545,1665,2077,2201,2325,2449,2573,2697,2821,2945,19065$,\\
    $19275,21845,26985,95325,2080895$, and numbers specified by
    Theorem~\ref{thm9}.
    \item[(i)] Every odd integer with eight $0$'s is flimsy, with the exception of $513,1027,1029,1057,1281$,\\
    $2055,2085,2089,2097,2115,2145,2193,2313,2337,2563,2565,2625,3075,3105,3585,4123$,\\
    $4185,4371,4389,4433,4619,4675,4681,4867,4929,5187,6169,6417,6665,6913,8253,8505$,\\
    $8525,8645,8757,9009,9261,9513,9765,10017,10269,10465,10521,10773,11025,11277$,\\
    $11529,11781,12033,12483,13505,14497,18631,25623,34695,39321,42405,50115,57825$,\\
    $158875,222425,774333,16711935$, and numbers specified by
    Theorem~\ref{thm9}.
    \item[(j)] Every odd integer with nine $0$'s is flimsy, with the exception of
    $1025,2051,2057,2065,2177$,\\
    $3073,4131,4165,4233,4361,4369,4417,4641,5129,5185,6273,8215,8277,8339,8401,8711$,\\
    $8773,8835,8897,10261,10385,10757,10881,12307,12369,12803,12865,14353,14849,16443$,\\
    $16569,16835,16947,17073,17451,17577,17745,17955,18081,18459,18585,18963,19089$,\\
    $19467,19593,19971,20097,24605,24633,25025,25137,25641,26145,26649,26691,27153$,\\
    $27657,28161,28679,32893,33401,33909,34417,34925,35433,35941,36449,36957,37465$,\\
    $37973,38481,38989,39497,40005,40513,41021,41529,41769,42037,42545,43053,43561$,\\
    $44069,44577,45085,45593,46101,46609,47117,47625,48133,48641,178481,285975,349525$,\\
    $413075,476625,1290555,1806777,1864135,6242685,133956095$, and numbers specified by
    Theorem~\ref{thm9}. \end{itemize}
\end{theorem}

Based on this theorem, we make the following
conjecture.
\begin{conjecture}
Every number with $j$ $0$'s
is flimsy, with exceptions of the form $s 1^i \overline{s}$, $i \geq 0$, where $|s| = j$ and $s$ begins with $1$ and
ends with $0$,
and only finitely many additional exceptions.
\end{conjecture}

%



\section{The $k$-flimsy numbers via formal language theory}
\label{flim}

In this section we describe a new approach, based on formal language theory, for understanding the distribution of the $k$-flimsy numbers.  Recall these are the numbers 
$$F_k = \{ n \geq  1 \ : \ s_2(kn) < s_2 (n) \}.$$   The majority of the results in this section are about the case $k = 3$, although in principle our technique can be applied to any odd $k$.

K{\'atai} \cite{Katai:1977} studied the difference $s_2(3n) - s_2 (n)$,
and proved that this quantity is essentially normally distributed,
in a certain sense.
Stolarsky \cite{Stolarsky:1980} conjectured that the natural density of the $k$-flimsy numbers is $1/2$ for all odd $k$.   His conjecture was later proved by W. M. Schmidt \cite{Schmidt:1983} and J. Schmid \cite{Schmid:1984}.  All
these results use rather sophisticated tools of number theory and probability.

In contrast, in this section
we obtain rather detailed results on the distribution of $3$-flimsy numbers through a (more or less) purely mechanical approach based on formal language
theory.  The main result of this section is the following:
\begin{theorem}
The number of $3$-flimsy numbers in the interval
$[2^{N-1}, 2^N)$ is
\begin{equation}
2^N \left({1 \over 4} - cN^{-1/2} + O(N^{-3/2}) \right),
\label{sturd}
\end{equation}
where $c = {{7\sqrt{6}} \over {24 \sqrt{\pi}}} \doteq 0.4030765$.
\label{clok}
\end{theorem}

Our method starts with a pushdown automaton (PDA) recognizing the $k$-flimsy numbers, and by a series of steps, it is converted into an 
asymptotic series expansion for the number of $k$-flimsy numbers with $N$ bits.   Previously, the basic approach has been used for a wide variety of combinatorial enumerations; see, for example,
\cite{Banderier&Drmota:2013,Banderier&Drmota:2015,Asinowski&Bacher&Banderier&Gittenberger:2018,Asinowski&Bacher&Banderier&Gittenberger:2019}.
We have implemented all the steps, and the flow of control is explained in the diagram below.
\begin{center}
    \includegraphics{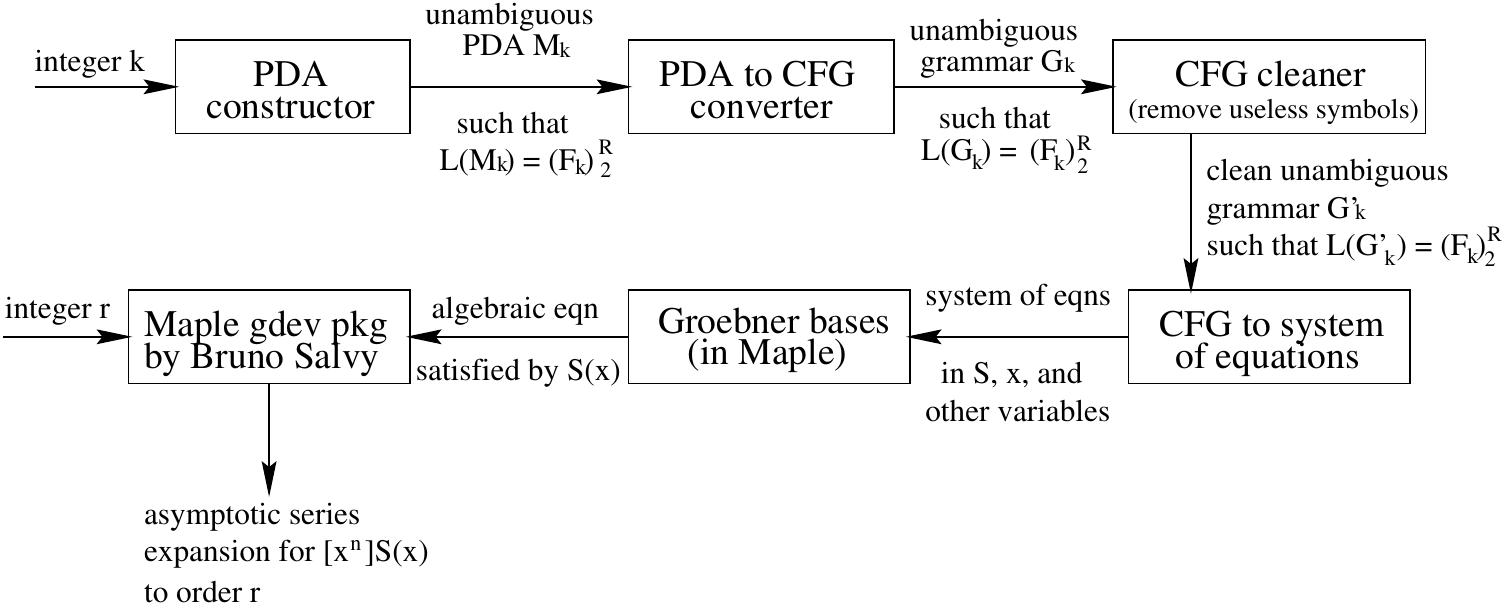}
\end{center}
We now explain briefly what each box in the diagram does, with more detailed explanation to follow.  For all undefined terms, see any textbook on automata theory or formal languages, such as \cite{Hopcroft&Ullman:1979}.

First, given an odd integer $k \geq 3$, we build an unambiguous pushdown automaton (PDA) $M_k$
that recognizes the base-$2$ representation of elements of
$F_k$; more precisely, $M_k$ recognizes the language $(F_k)_2^R$.  The length-$N$ strings in $(F_k)_2^R$ are in 1-1 correspondence with the flimsy numbers in the half-open interval
$[2^{N-1}, 2^N)$, so our goal is to estimate the cardinality of
$(F_k)_2^R \cap  \{ 0, 1\}^{N}$ as precisely as possible.

Second, we convert $M_k$ to an unambiguous context-free grammar $G_k$
generating $(F_k)_2^R$.

We simplify this context-free grammar by deleting useless symbols (those symbols that do not participate in the derivation of any terminal string, or are not reachable from the start variable), obtaining a new CFG $G'_k$.

Third, we convert $G'_k$ to a system of equations in the variables of $G'_k$.  These variables represent formal power series, with the property that the number of length-$N$ strings generated by a variable $A$ is given
by $[x^N] A(x)$, the coefficient of $x^N$ in the power series $A$.

Fourth, using Gr\"obner bases, we solve this system of equations, obtaining an algebraic equation satisfied by the formal power series $S(x)$, where $S$ is the start variable of the grammar $G'_k$.

Finally, using Bruno Salvy's {\tt gdev} package, written in Maple, we can determine the asymptotic behavior of $[x^N]S(x)$ using the saddle-point method (as discussed by, e.g., Flajolet and Sedgewick \cite{Flajolet&Sedgewick:2009}).  In principle, we can obtain as many terms as we wish of the asymptotic expansion.  

Theorem~\ref{clok} now follows by performing each of these steps.  The first four steps are done with original code written by the first author in Python, and the last two steps are done with Maple.  The code for each step is available at \\
\centerline{{\tt  https://github.com/FinnLidbetter/sturdy-numbers } \ \ .}
We now give more complete details of some of the steps.

\subsection{Constructing the PDA $M_k$}

The general idea is as follows:  we create a PDA accepting
the base-2 representation of $k$-flimsy numbers $n$.  We use the
stack of the PDA to record the absolute value of
$s_2 (n) - s_2 (kn)$, and we use the state to record
both the carry needed when multiplying input by $k$,
and the sign of $s_2 (n) - s_2 (kn)$. We accept the
input if the carry is 0, the sign of $s_2 (n) - s_2 (kn)$
is positive, and the stack has at least one counter.



Our PDA is assumed to begin its computation with a special symbol, {\tt Z}, on
top of the stack, and if the input is accepted, to end its computation when the stack becomes empty.

The sketch above is not quite enough because of two technical
issues.  First, (a) in some cases this approach
requires reading
extra leading zeroes (which, because we are representing
numbers starting with the least significant digit first,
would be at the end of the input), in order to guarantee
that the carry for $s_2(kn)$ was taken into account and
(b) we must have that the leading bit of the input is 1, to avoid
incorrectly counting smaller numbers as having $n$ bits.

To handle both these issues, we slightly modify the construction
in several ways.
First, if the state has a minus sign, then the stack
holds $|y|_1 - |x|_1$ {\tt X}'s, where $x$ is the input seen so far and
$y$ is the $|x|$ least significant bits of $k (x)^R_2$.
On the other hand, if the state has a positive sign, then
the stack holds $|x|_1 - |y|_1 - 1$ {\tt X}'s.

Second, to simulate the needed leading zeroes required to handle
the carry, without actually reading them, we use a special series of
$\log_2 k$ states to pop {\tt X}'s from the stack.

Finally, we have a special state used to empty the stack when acceptance is detected.  The total number of states is therefore at most $2k + \log_2 k$.

The resulting PDA $M_3$ is depicted
in Figure~\ref{fig1}.
\begin{figure}[H]
\begin{center}
\includegraphics[width=3in]{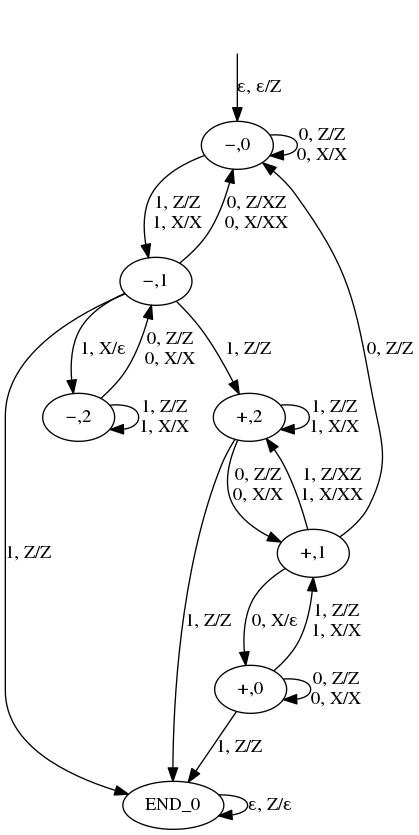}
\end{center}
\label{fig1}
\caption{PDA $M_3$.}
\end{figure}

One important property of our construction is that our PDA $M_k$ is {\it unambiguous}.   By this we mean that every accepted word has exactly one accepting computational path.

\subsection{Converting the PDA to a CFG}

We can convert $M_k$ to an equivalent context-free grammar $G_k$ using a standard technique
called the ``triple construction'' \cite[pp.~115--119]{Hopcroft&Ullman:1979}.   This gives us a grammar $G_k$ with $O(k^2)$ variables and $O(k^3)$ productions.


Now we use the fact, proved in
\cite[Thm.~5.4.3, p.~151]{Harrison:1978}, that performing the triple construction on an unambiguous PDA gives us an unambiguous grammar.  


\subsection{Cleaning the CFG}

We can remove useless symbols from our grammar $G_k$ by removing all variables that
do not derive a terminal string, then removing all productions containing these removed
variables, and then removing all variables and terminals that are not reachable from
the start variable.
This is a standard procedure, and is described in greater detail in \cite[pp.~88--90]{Hopcroft&Ullman:1979}.

Once this is complete, it may be found that there is a variable $X$ that has only one production $X \rightarrow \alpha$. If $X$ is not the start variable, then it can be deleted from the set of variables, and all instances of $X$ in production rules can be replaced with $\alpha$.

For example, when we convert our PDA $M_3$, we get an unambiguous grammar $G_3$; cleaning $G_3$ using this procedure gives us the following grammar $G'_3$:
\begin{align*}
S   & \rightarrow    1F  \orr  0S  & \quad
A   & \rightarrow    1E  \orr  0A \\
B   & \rightarrow    1G  \orr  0B & \quad 
C   & \rightarrow    1H  \orr  1  \orr  0C \\
D   & \rightarrow    1I  \orr  0D & \quad
E   & \rightarrow    1   \orr  0AJ \\
F   & \rightarrow    1N  \orr  0AK & \quad
G   & \rightarrow    1LB \orr  0 \\
H   & \rightarrow    1M  \orr  1LC  \orr  1 & \quad
I   & \rightarrow    1M  \orr  1LD  \orr  1  \orr  0S \\
J   & \rightarrow    1J  \orr  0E & \quad
K   & \rightarrow    1K  \orr  0F \\
L   & \rightarrow    1L  \orr  0G & \quad
M   & \rightarrow    1M  \orr  1  \orr  0H \\
N   & \rightarrow    1N  \orr  0I
\end{align*}

\subsection{Converting the CFG to a system of equations}

This transformation was discussed in \cite{Chomsky&Schutzenberger:1963}.   It suffices
to replace, in each set of productions $A \rightarrow \alpha_1 \ | \ \alpha_2 \ | \ \ldots \ | \ \alpha_i$
of a grammar $G$, each terminal symbol by the indeterminate $x$, each $|$ symbol by a
plus sign, and the $\rightarrow$ with an equals sign.
For a proof of correctness, see \cite{Kuich&Salomaa:1986,Panholzer:2005}.

Performing this transformation on $G'_3$ gives us the following system of equations:
\begin{align*}
S   & =    xF  +  xS  & \quad 
A   & =    xE  +  xA \\
B   & =    xG  +  xB & \quad 
C   & =    xH  +  x  +  xC \\
D   & =    xI  +  xD & \quad 
E   & =    x   +  xAJ \\
F   & =    xN  +  xAK & \quad
G   & =    xLB +  x \\
H   & =    xM  +  xLC  +  x & \quad
I   & =    xM  +  xLD  +  x  +  xS \\
J   & =    xJ  +  xE & \quad
K   & =    xK  +  xF \\
L   & =    xL  +  xG & \quad
M   & =    xM  +  x  +  xH \\
N   & =    xN  +  xI
\end{align*}

\subsection{Solving the system}
\label{maple_groebner}

We can now solve the resulting system of equations for $S$,
obtaining an algebraic equation for which $S$ is the root.
The main tool is Groebner bases, for which a helpful package already exists in {\tt Maple}.

Using the code given in Appendix~\ref{mapleapp}, we find the following quadratic equation for $S$ in the case $k = 3$.
$$ x (2 x - 1)^2 (x + 1) (2 x^2 - x + 1) S(x)^2   +   (2 x - 1) (x - 1)^2 (x + 1) (2 x^2 - x + 1) S(x)   +   x^4 (x^2 - x + 1) = 0.$$

Solving this quadratic for $S$ gives
$$S(x) = {{ -(x-1)^2(x+1)(2x^2-x+1) + \sqrt{ -(x-1)(2x-1)(2x^2-x+1)(x^3+x^2-x+1)^2 }}  \over { 2x(2x-1)(x+1)(2x^2-x+1) }}.$$
Since the grammar $G'_3$ is unambiguous, the formal power series $S(x)$ is the census generating function for the set $(F_3)^R_2$. In particular, this means that $[x^N]S(x) = |F_3 \cap [2^{N-1},2^N)|$, or in other words, the coefficient of $x^N$ in $S(x)$ is the number $k$-flimsy numbers in $[2^{N-1},2^N)$.

\subsection{Asymptotic expansion of the coefficients of the power series}
\label{maple_asymp}

Finally, we use Flajolet-Sedgewick-style asymptotic analysis
\cite[\S VII. 7.1]{Flajolet&Sedgewick:2009} to determine an asymptotic formula for the $N$'th coefficient of the power series expansion for $S(x)$.  Conveniently, there is a {\tt Maple} package {\tt algolib}, written by Bruno Salvy \cite{Salvy:2013}, to accomplish this.   When we run this on our formula for $S(x)$, we get our desired result.

This completes our discussion of the proof of Theorem~\ref{clok}.

\begin{remark}
We could easily determine more terms in the asymptotic expansion, if we wanted, using the same ideas.  For example, we can find that the number of $3$-flimsy numbers in the interval $[2^{N-1}, 2^N)$ is
$$ 2^N \left({1 \over 4} - 
{{\sqrt{6}} \over \sqrt{\pi}} \left( {7 \over {24}}  N^{-1/2} + {{13} \over {72}} N^{-3/2}  - {{17} \over {64}} N^{-5/2} + 
{{3365} \over {13824}} N^{-7/2}
+ \cdots  \right) \right).$$
\end{remark}

\begin{corollary} 
The number of $3$-flimsy
numbers $< 2^N$ is 
$2^{N-1}  - O(2^N N^{-1/2})$.
\end{corollary}

\begin{proof}
For any real number $a > 0$ we have
\begin{align*}
2^N N^{-a} \leq \sum_{1 \leq n \leq N} 2^n n^{-a} & \leq \sum_{1 \leq n \leq N/2} 2^n n^{-a} + \sum_{N/2 < n \leq N} 2^n n^{-a} \\
& \leq \sum_{1 \leq n \leq N/2} 2^n + (N/2)^{-a} \sum_{N/2 < n \leq N} 2^n \\
& \leq 2^{N/2+1} + (N/2)^{-a} 2^{N+1} .
\end{align*}
Summing \eqref{sturd} and applying the inequalities above
gives the desired result.
\end{proof}

\begin{theorem}
The number of 5-flimsy numbers in the interval $[2^{N-1}, 2^N)$ is 
\begin{equation}
2^N \left({1 \over 4} - cN^{-1/2} + O(N^{-3/2}) \right),
\label{5sturd}
\end{equation}
where $c = {{3\sqrt{5}} \over {8 \sqrt{\pi}}} \doteq 0.473087348$.
\end{theorem}
\begin{proof}
This is determined using the same method as the proof for Theorem~\ref{clok}.
In this process, we compute the unambiguous PDA $M_5$, depicted in Figure~\ref{fig2}.
\begin{figure}[H]
\begin{center}
\includegraphics[width=4in]{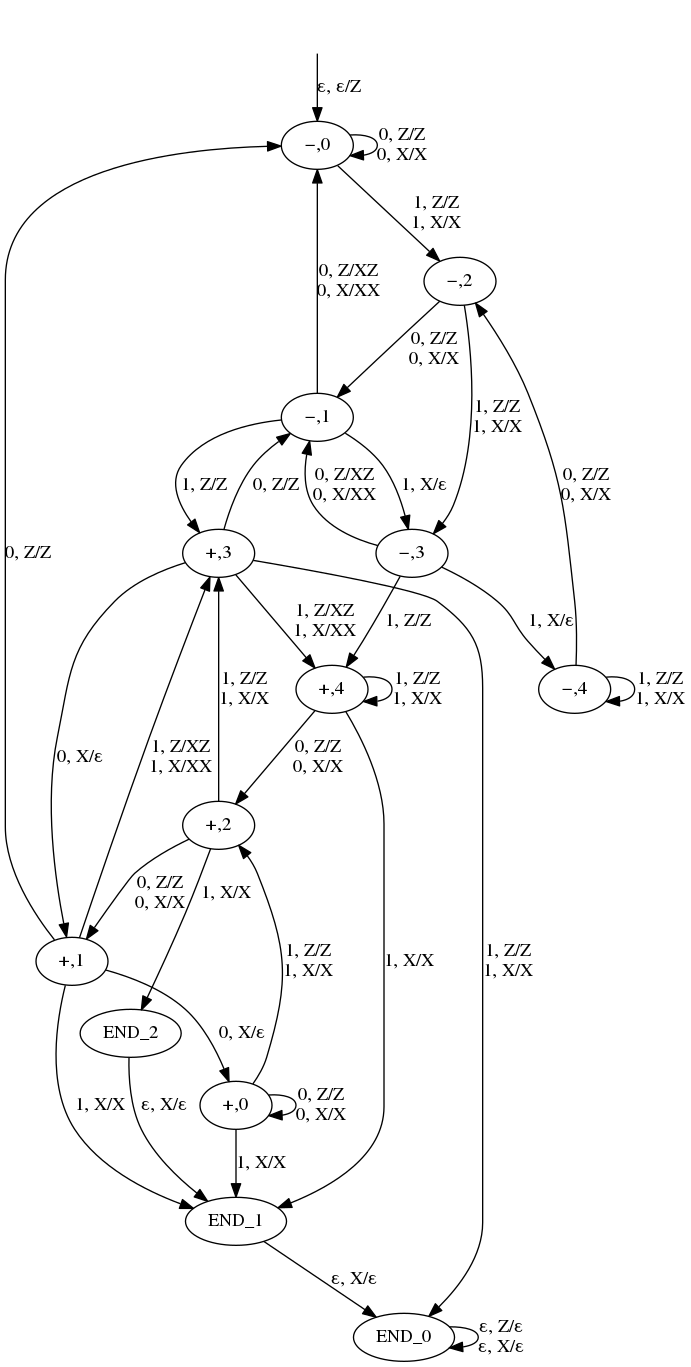}
\end{center}
\label{fig2}
\caption{PDA $M_5$.}
\end{figure}

From $M_5$ we construct the (simplified) unambiguous grammar $G'_5$ as follows:
\begin{align*}
    S	    &\rightarrow	1 V_1 \orr 0 S & \quad
    V_1	    &\rightarrow	1 V_{29} \orr 0 V_{27}\\
    V_2	    &\rightarrow	1 V_{36} \orr 1 \orr 0 V_2 & \quad
    V_3	    &\rightarrow	1 V_3 \orr 0 V_{12}\\
    V_4	    &\rightarrow	1 V_4 \orr 1 \orr 0 V_{36} & \quad
    V_5	    &\rightarrow	1 V_{10} V_{37} \orr 1 V_{20} V_2 \orr 1 V_{24} \orr 1 V_4 \orr 1\\
    V_6	    &\rightarrow	1 V_{28} \orr 0 V_6 & \quad
    V_7	    &\rightarrow	1 V_{10} V_{21} \orr 1 V_{20} V_{11} \orr 0\\
    V_8	    &\rightarrow	1 V_7 V_8 \orr 1 V_{22} V_{25} & \quad
    V_9	    &\rightarrow	1 V_7 V_9 \orr 1 V_{22} V_{19} \orr 1 V_5 \orr 0 S\\
    V_{10}	&\rightarrow	1 V_{10} \orr 0 V_{23} & \quad
    V_{11}	&\rightarrow	1 V_{23} \orr 0 V_{11}\\
    V_{12}	&\rightarrow	1 V_{14} \orr 0 V_9 & \quad
    V_{13}	&\rightarrow	1 V_{13} \orr 0 V_{39}\\
    V_{14}	&\rightarrow	1 V_{10} V_9 \orr 1 V_{20} V_{19} \orr 1 V_4 \orr 1 \orr 0 V_{27} & \quad
    V_{15}	&\rightarrow	1 V_7 V_{15} \orr 1 V_{22} V_{26} \orr 0\\
    V_{16}	&\rightarrow	1 V_{10} V_8 \orr 1 V_{20} V_{25} & \quad
    V_{17}	&\rightarrow	0 V_6 V_{33} \orr 0 V_{18} V_{34}\\
    V_{18}	&\rightarrow	1 V_{39} \orr 0 V_{18} & \quad
    V_{19}	&\rightarrow	1 V_{12} \orr 0 V_{19}\\
    V_{20}	&\rightarrow	1 V_{20} \orr 0 V_{31} & \quad
    V_{21}	&\rightarrow	1 V_7 V_{21} \orr 1 V_{22} V_{11}\\
    V_{22}	&\rightarrow	1 V_{10} V_{15} \orr 1 V_{20} V_{26} & \quad
    V_{23}	&\rightarrow	1 V_7 \orr 0 V_{21}\\
    V_{24}	&\rightarrow	1 V_{24} \orr 0 V_{35} & \quad
    V_{25}	&\rightarrow	1 V_{35} \orr 0 V_{25}\\
    V_{26}	&\rightarrow	1 V_{31} \orr 0 V_{26} & \quad
    V_{27}	&\rightarrow	1 V_{14} \orr 0 V_6 V_{30} \orr 0 V_{18} V_{29}\\
    V_{28}	&\rightarrow	1 V_{34} \orr 0 V_{17} & \quad
    V_{29}	&\rightarrow	1 V_3 \orr 0 V_{17} V_{30} \orr 0 V_{32} V_{29}\\
    V_{30}	&\rightarrow	1 V_{30} \orr 0 V_1 & \quad
    V_{31}	&\rightarrow	1 V_{22} \orr 0 V_{15}\\
    V_{32}	&\rightarrow	1 \orr 0 V_6 V_{13} \orr 0 V_{18} V_{38} & \quad
    V_{33}	&\rightarrow	1 V_{33} \orr 0 V_{28}\\
    V_{34}	&\rightarrow	1 \orr 0 V_{17} V_{33} \orr 0 V_{32} V_{34} & \quad
    V_{35}	&\rightarrow	1 V_{16} \orr 1 \orr 0 V_8\\
    V_{36}	&\rightarrow	1 V_5 \orr 0 V_{37} & \quad
    V_{37}	&\rightarrow	1 V_7 V_{37} \orr 1 V_{22} V_2 \orr 1 V_{16} \orr 1 V_5 \orr 1\\
    V_{38}	&\rightarrow	0 V_{17} V_{13} \orr 0 V_{32} V_{38} & \quad
    V_{39}	&\rightarrow	1 V_{38} \orr 0 V_{32}
\end{align*}

This gives us the system of equations:
\begin{align*}
    S	     &=    x  V_1 + x S &\quad
    V_1	     &=    x  V_{29} + x V_{27}\\
    V_2	     &=    x  V_{36} + x + x V_2 &\quad
    V_3	     &=    x  V_3 + x V_{12}\\
    V_4	     &=    x  V_4 + x + x V_{36} &\quad
    V_5	     &=    x  V_{10} V_{37} + x V_{20} V_2 + x V_{24} + x V_4 + x\\
    V_6	     &=    x  V_{28} + x V_6 &\quad
    V_7	     &=    x  V_{10} V_{21} + x V_{20} V_{11} + x\\
    V_8	     &=    x  V_7 V_8 + x V_{22} V_{25} &\quad
    V_9	     &=    x  V_7 V_9 + x V_{22} V_{19} + x V_5 + x S\\
    V_{10}	 &=    x  V_{10} + x V_{23} &\quad
    V_{11}	 &=    x  V_{23} + x V_{11}\\
    V_{12}	 &=    x  V_{14} + x V_9 &\quad
    V_{13}	 &=    x  V_{13} + x V_{39}\\
    V_{14}	 &=    x  V_{10} V_9 + x V_{20} V_{19} + x V_4 + x + x V_{27} &\quad
    V_{15}	 &=    x  V_7 V_{15} + x V_{22} V_{26} + x\\
    V_{16}	 &=    x  V_{10} V_8 + x V_{20} V_{25} &\quad
    V_{17}	 &=    x  V_6 V_{33} + x V_{18} V_{34}\\
    V_{18}	 &=    x  V_{39} + x V_{18} &\quad
    V_{19}	 &=    x  V_{12} + x V_{19}\\
    V_{20}	 &=    x  V_{20} + x V_{31} &\quad
    V_{21}	 &=    x  V_7 V_{21} + x V_{22} V_{11}\\
    V_{22}	 &=    x  V_{10} V_{15} + x V_{20} V_{26} &\quad
    V_{23}	 &=    x  V_7 + x V_{21}\\
    V_{24}	 &=    x  V_{24} + x V_{35} &\quad
    V_{25}	 &=    x  V_{35} + x V_{25}\\
    V_{26}	 &=    x  V_{31} + x V_{26} &\quad
    V_{27}	 &=    x  V_{14} + x V_6 V_{30} + x V_{18} V_{29}\\
    V_{28}	 &=    x  V_{34} + x V_{17} &\quad
    V_{29}	 &=    x  V_3 + x V_{17} V_{30} + x V_{32} V_{29}\\
    V_{30}	 &=    x  V_{30} + x V_1 &\quad
    V_{31}	 &=    x  V_{22} + x V_{15}\\
    V_{32}	 &=    x  + x V_6 V_{13} + x V_{18} V_{38} &\quad
    V_{33}	 &=    x  V_{33} + x V_{28}\\
    V_{34}	 &=    x  + x V_{17} V_{33} + x V_{32} V_{34} &\quad
    V_{35}	 &=    x  V_{16} + x + x V_8\\
    V_{36}	 &=    x  V_5 + x V_{37} &\quad
    V_{37}	 &=    x  V_7 V_{37} + x V_{22} V_2 + x V_{16} + x V_5 + x\\
    V_{38}	 &=    x  V_{17} V_{13} + x V_{32} V_{38} &\quad
    V_{39}	 &=    x  V_{38} + x V_{32}
\end{align*}
\end{proof}

Finally, we employ the same methods and {\tt Maple} packages used in
Sections~\ref{maple_groebner} and~\ref{maple_asymp}, which give us
our desired result.

We can also use the same ideas to compute the distribution of flimsy numbers in other bases.  As an example we proved
\begin{theorem}
The number of integers in the range $[3^{N-1}, 3^N)$ that
are $2$-flimsy in base $3$ is $$
3^N \left({1 \over 3} + {{\sqrt{3}} \over \sqrt{\pi N}} \left( -{1 \over 3} + {1 \over {48 N}} - {{13} \over {1536 N^2}} - {{65} \over {24576 N^3}}  + O(N^{-4}) \right) \right). $$
\end{theorem}

\begin{proof}
As before.  We omit the details.
\end{proof}



\section{The $k$-equal numbers via formal language theory}
\label{kequ}

Another quantity of interest is the number of $n$
for which $s_2 (n) = s_2 (kn)$.  
We call such $n$ {\it $k$-equal}.   By generalizing
the approach used in Section~\ref{flim}, we can compute how many integers
$n\in[2^{N-1}, 2^N)$ are $k$-equal.

In particular, we modify PDA $M_k$ by changing the transitions to the {\tt END}
states. Whereas $M_k$ transitions to {\tt END} when reading a $1$ if following that
$1$ with sufficiently many zeros would reach the state $(+,0)$,
instead we want such an input to reach the state $(-,0)$ with no counters
on the stack.  

This modification gives us the following unambiguous PDAs $M'_k$ that recognize the
$k$-equal numbers for $k=3,5$:
\begin{figure}[H]
\begin{center}
\includegraphics[width=3in]{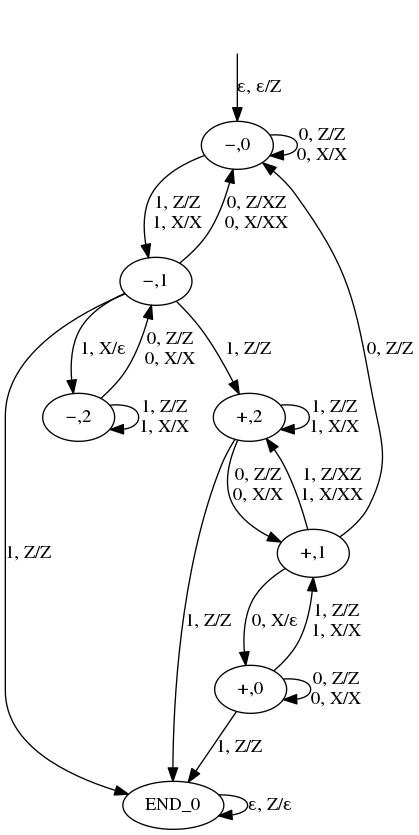}
\end{center}
\label{fig9}
\caption{PDA $M'_3$ recognizing 3-equal numbers.}
\end{figure}

\begin{figure}[H]
\begin{center}
\includegraphics[width=4in]{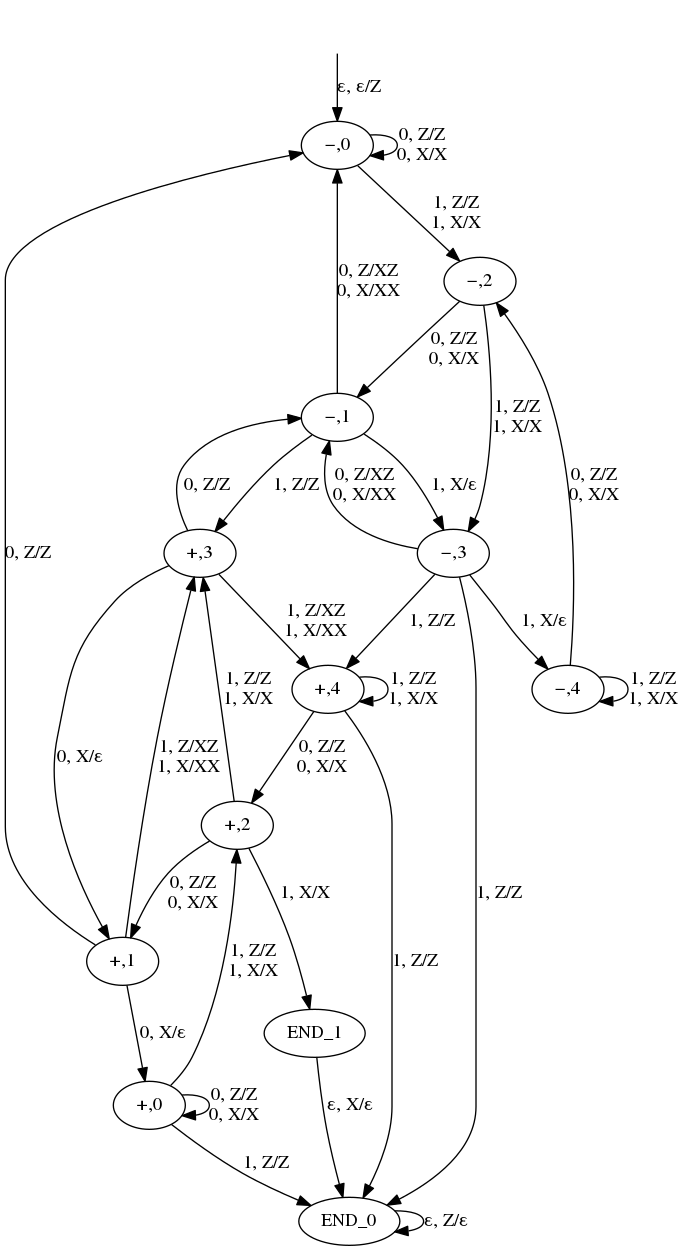}
\end{center}
\label{fig10}
\caption{PDA $M'_5$ recognizing the 5-equal numbers.}
\end{figure}

Following the same procedure described in Theorem~\ref{clok} of converting
$M'_3$ and $M'_5$ into unambiguous CFGs and simplifying them gives us the following:

The unambiguous CFG for 3-equal numbers:
\begin{align*}
S      & \rightarrow	0 S \orr 1 V_6 & \quad
V_1    & \rightarrow	0 V_{11} \orr 1 V_1 \\
V_2    & \rightarrow	0 V_5 V_{10} \orr 1 & \quad
V_3    & \rightarrow	0 V_3 \orr 1 V_{11} \\
V_4    & \rightarrow	0 V_4 \orr 1 V_7 \orr 1 & \quad
V_5    & \rightarrow	0 V_5 \orr 1 V_2 \\
V_6    & \rightarrow	0 V_5 V_9 \orr 1 V_8 \orr 1 & \quad
V_7    & \rightarrow	0 S \orr 1 V_1 V_4 \\
V_8    & \rightarrow	0 V_7 \orr 1 V_8 \orr 1 & \quad
V_9    & \rightarrow	0 V_6 \orr 1 V_9 \\
V_{10} & \rightarrow	0 V_2 \orr 1 V_{10} & \quad
V_{11} & \rightarrow	0 \orr 1 V_1 V_3
\end{align*}

The unambiguous CFG for 5-equal numbers:
\begin{align*}
S	& \rightarrow	0 S \orr 1 V_{29} & \quad
V_1	& \rightarrow	0 V_{32} \orr 1 V_{33} \\
V_2	& \rightarrow	0 V_5 V_8 \orr 0 V_{14} V_2 \orr 1 & \quad
V_3	& \rightarrow	1 V_{33} V_{34} \orr 1 V_{21} V_3 \\
V_4	& \rightarrow	0 V_4 \orr 1 V_{23} & \quad
V_5	& \rightarrow	0 V_{30} V_8 \orr 0 V_4 V_2 \\
V_6	& \rightarrow	1 V_{27} V_{34} \orr 1 V_{28} V_3 & \quad
V_7	& \rightarrow	0 V_7 \orr 1 V_1 \\
V_8	& \rightarrow	0 V_{15} \orr 1 V_8 & \quad
V_9	& \rightarrow	0 V_9 \orr 1 V_{26} \orr 1 \\
V_{10}	& \rightarrow	0 V_{10} \orr 1 V_{19} & \quad
V_{11}	& \rightarrow	0 V_{30} V_{17} \orr 0 V_4 V_{12} \orr 1 V_{22} \\
V_{12}	& \rightarrow	0 V_5 V_{17} \orr 0 V_{14} V_{12} \orr 1 V_{31} \orr 1 & \quad
V_{13}	& \rightarrow	1 V_{33} V_{10} \orr 1 V_{21} V_{13} \\
V_{14}	& \rightarrow	0 V_{30} V_{24} \orr 0 V_4 V_{18} \orr 1 & \quad
V_{15}	& \rightarrow	0 V_5 \orr 1 V_2 \\
V_{16}	& \rightarrow	0 S \orr 1 V_{33} V_9 \orr 1 V_6 \orr 1 V_{21} V_{16} \orr 1 & \quad
V_{17}	& \rightarrow	0 V_{29} \orr 1 V_{17} \\
V_{18}	& \rightarrow	0 V_5 V_{24} \orr 0 V_{14} V_{18} & \quad
V_{19}	& \rightarrow	0 V_{13} \orr 1 V_{21} \\
V_{20}	& \rightarrow	0 V_{25} \orr 1 V_{20} & \quad
V_{21}	& \rightarrow	0 \orr 1 V_{27} V_{10} \orr 1 V_{28} V_{13} \\
V_{22}	& \rightarrow	0 V_{11} \orr 1 V_{27} V_9 \orr 1 V_{20} \orr 1 V_{28} V_{16} & \quad
V_{23}	& \rightarrow	0 V_{14} \orr 1 V_{18} \\
V_{24}	& \rightarrow	0 V_{23} \orr 1 V_{24} & \quad
V_{25}	& \rightarrow	0 V_3 \orr 1 V_6 \orr 1 \\
V_{26}	& \rightarrow	0 V_{16} \orr 1 V_{22} & \quad
V_{27}	& \rightarrow	0 V_1 \orr 1 V_{27} \\
V_{28}	& \rightarrow	0 V_{19} \orr 1 V_{28} & \quad
V_{29}	& \rightarrow	0 V_{11} \orr 1 V_{12} \\
V_{30}	& \rightarrow	0 V_{30} \orr 1 V_{15} & \quad
V_{31}	& \rightarrow	0 V_{26} \orr 1 V_{31} \orr 1 \\
V_{32}	& \rightarrow	0 \orr 1 V_{33} V_7 \orr 1 V_{21} V_{32} & \quad
V_{33}	& \rightarrow	1 V_{27} V_7 \orr 1 V_{28} V_{32} \\
V_{34}	& \rightarrow	0 V_{34} \orr 1 V_{25}
\end{align*}

Finally, converting these grammars into systems of equations, solving, and
finding the asymptotics gives us the following results.

\begin{theorem}
The number of $3$-equal numbers in the interval
$[2^{N-1}, 2^N)$ is
\begin{equation}
2^N \left(cN^{-1/2}  + O(N^{-3/2}) \right),
\label{sturd3}
\end{equation}
where $c = {{\sqrt{6}} \over {4\sqrt{\pi}}} \doteq 0.345494149$.
\label{3equal}
\end{theorem}

\begin{theorem}
The number of $5$-equal numbers in the interval
$[2^{N-1}, 2^N)$ is
\begin{equation}
2^N \left(cN^{-1/2}  + O(N^{-3/2}) \right),
\label{sturd5}
\end{equation}
where $c = {{\sqrt{5}} \over {4\sqrt{\pi}}} \doteq 0.315391565$.
\label{5equal}
\end{theorem}

\section{Conclusions and open problems}
We have shown that techniques from automata theory can be used to solve problems in number theory.   For other fun along these lines,
see \cite{Bell&Hare&Shallit:2019,Rajasekaran&Shallit&Smith:2019}.

It would be interesting to understand the distribution of values of $\msw(n)$ and
$\mfw(n)$ for $n$ flimsy.  We leave this as an open problem.

\section{Acknowledgments}

We thank Kenneth Stolarsky for helpful comments.

\appendix

\section{Appendix:  the dynamic programming algorithm}
\label{appa}

In the pseudocode that follows, the scope of loops is indicated by the indentation.
\begin{verbatim}
minrep(n)   { assumes n odd and at least 3 }

sumd := sumdig(2,n);   {sum of base-2 digits of n}

{make a table of powers of 2}
b := 1;
a := 0;
repeat
   b := (2*b) mod n;
   a := a+1;
until
   b = 1;
ord2 := a;   { the order of 2 mod n }
power2 := array[0..ord2-1] of integer;
for m := 0 to ord2-1 do
   power2[m] := b;
   b := (2*b) mod n;

{ the intent is that x[i,j,r] = true, if j (mod n) has a representation 
as a sum of exactly i powers of 2, using only the first r elements of 
power2 (without repetition), and false otherwise.  
y[i,j,r] = smallest integer congruent to j (mod n)
representable by the sum of exactly i powers of 2,
using only the first r elements of power2 (without repetition) }

x := array[0..sumd-1, 0..n-1, 0..ord2] of boolean;
y := array[0..sumd-1, 0..n-1, 0..ord2] of integer;

{ initialize }

for r := 0 to ord2 do
     for i := 1 to sumd-1 do
         for j := 0 to n-1 do
            x[i,j,r] := false;
            y[i,j,r] := infinity;
    x[0,0,r] := true;
    y[0,0,r] := 0;
    
{ fill in table }

for r := 1 to ord2 do   {consider summand 2^{r-1} mod p}
   for j := 0 to n-1 do      { check position j }
       for i := 1 to sumd-1 do     {fill in level i of the array}
            x[i,j,r] := x[i,j,r-1];
            y[i,j,r] := y[i,j,r-1];
            {check if we can use 2^{r-1} }
            if x[i-1, (j-power2[r-1]) mod n, r-1] then
                x[i,j,r] := true;
                y[i,j,r] := min(y[i,j,r], 
                y[i-1, (j-power2[r-1]) mod n, r-1] + 2^{r-1});

sturdy := true;
for i := 2 to sumd-1 do
    sturdy := sturdy and x[i,0,ord2];

if (sturdy) then
    print("swm(n) = ",sumd);
    print("msw(n) = ",1);
else
    i := 1;
    while (not x[i,0,ord2]) do i := i+1;
    print("swm(n) = ",i);
    print("msw(n) = ",y[i,0,ord2]/n);
    mfw := infinity;
    while (i < sumd) do 
        mfw := min(mfw, y[i,0,ord2]);
        i := i+1;
    print("mfw(n) = ",mfw);

end;
\end{verbatim}

\section{Appendix:  Maple code}
\label{mapleapp}
To use this code, you will first need to download the {\tt algolib} package from \url{http://algo.inria.fr/libraries/}.

\begin{verbatim}
eqs := [-S + x*V_F + x*S,
-V_A + x*V_E + x*V_A,
-V_B + x*V_G + x*V_B,
-V_C + x*V_H + x + x*V_C,
-V_D + x*V_I + x*V_D,
-V_E + x + x*V_A*V_J,
-V_F + x*V_N + x*V_A*V_K,
-V_G + x*V_L*V_B + x,
-V_H + x*V_M + x*V_L*V_C + x,
-V_I + x*V_M + x*V_L*V_D + x + x*S,
-V_J + x*V_J + x*V_E,
-V_K + x*V_K + x*V_F,
-V_L + x*V_L + x*V_G,
-V_M + x*V_M + x + x*V_H,
-V_N + x*V_N + x*V_I]:
Groebner[Basis](eqs, lexdeg([V_A, V_B, V_C, V_D, V_E, V_F, V_G, V_H, 
V_I, V_J, V_K, V_L, V_M, V_N], [S]));
algeq := %[1]:
map(series, [solve(algeq, S)], x);
f := solve(algeq,S);
ps := f[1]:
assume(x, positive):
series(ps, x, 40);
libname:="<insert current directory path>",libname:
combine(equivalent(ps,x,n,5));
\end{verbatim}

\section{Table of values}

Here the column labeled ``char''
is F if the number is flimsy and
S if it is sturdy.

\begin{table}
\centering
\resizebox{\columnwidth}{!}{%
\begin{tabular}{ccccc||ccccc}
$n$ & char & $\swm(n)$ & $\msw(n)$ & $\mfw(n)$ & 
$n$ & char & $\swm(n)$ & $\msw(n)$ & $\mfw(n)$ \\
\hline
3 & S & 2 & 1 & -  & 5 & S & 2 & 1 & -  \\ 
7 & S & 3 & 1 & -  & 9 & S & 2 & 1 & -  \\ 
11 & F & 2 & 3 & 3 & 13 & F & 2 & 5 & 5 \\ 
15 & S & 4 & 1 & -  & 17 & S & 2 & 1 & -  \\ 
19 & F & 2 & 27 & 27 & 21 & S & 3 & 1 & -  \\ 
23 & F & 3 & 3 & 3 & 25 & F & 2 & 41 & 41 \\ 
27 & F & 2 & 19 & 3 & 29 & F & 2 & 565 & 5 \\ 
31 & S & 5 & 1 & -  & 33 & S & 2 & 1 & -  \\ 
35 & S & 3 & 1 & -  & 37 & F & 2 & 7085 & 7085 \\ 
39 & F & 3 & 7 & 7 & 41 & F & 2 & 25 & 25 \\ 
43 & F & 2 & 3 & 3 & 45 & S & 4 & 1 & -  \\ 
47 & F & 3 & 11 & 3 & 49 & S & 3 & 1 & -  \\ 
51 & S & 4 & 1 & -  & 53 & F & 2 & 1266205 & 5 \\ 
55 & F & 3 & 7 & 3 & 57 & F & 2 & 9 & 9 \\ 
59 & F & 2 & 9099507 & 3 & 61 & F & 2 & 17602325 & 5 \\ 
63 & S & 6 & 1 & -  & 65 & S & 2 & 1 & -  \\ 
67 & F & 2 & 128207979 & 128207979 & 69 & S & 3 & 1 & -  \\ 
71 & F & 3 & 119 & 119 & 73 & S & 3 & 1 & -  \\ 
75 & S & 4 & 1 & -  & 77 & F & 3 & 5 & 5 \\ 
79 & F & 3 & 13 & 7 & 81 & F & 2 & 1657009 & 1657009 \\ 
83 & F & 2 & 26494256091 & 395 & 85 & S & 4 & 1 & -  \\ 
87 & F & 3 & 3 & 3 & 89 & S & 4 & 1 & -  \\ 
91 & F & 3 & 3 & 3 & 93 & S & 5 & 1 & -  \\ 
95 & F & 3 & 5519 & 3 & 97 & F & 2 & 172961 & 172961 \\ 
99 & F & 2 & 331 & 11 & 101 & F & 2 & 11147523830125 & 365 \\ 
103 & F & 3 & 5 & 5 & 105 & S & 4 & 1 & -  \\ 
107 & F & 2 & 84179432287299 & 3 & 109 & F & 2 & 2405 & 5 \\ 
111 & F & 3 & 591 & 3 & 113 & F & 2 & 145 & 145 \\ 
115 & F & 3 & 571 & 9 & 117 & F & 4 & 5 & 5 \\ 
119 & F & 3 & 71 & 3 & 121 & F & 2 & 297758653049289 & 9 \\ 
123 & F & 4 & 19 & 3 & 125 & F & 2 & 9007199254741 & 5 \\ 
127 & S & 7 & 1 & -  & 129 & S & 2 & 1 & -  \\ 
\end{tabular}
}
\end{table}

\end{document}